\documentclass[notitlepage,onecolumn]{article}%
\usepackage{amsfonts}
\usepackage{graphicx}
\usepackage{geometry}
\usepackage{amsmath}
\usepackage{revsymb}
\usepackage{wrapfig}
\usepackage{amssymb}
\usepackage{hyperref}
\usepackage[square,numbers,sort&compress]{natbib}
\usepackage{hypernat}%
\providecommand{\U}[1]{\protect\rule{.1in}{.1in}}
\newtheorem{theorem}{Theorem}

\newtheorem{definition}[theorem]{Definition}

\newenvironment{proof}[1][Proof]{\textbf{#1.} }{\ \rule{0.5em}{0.5em}}
\begin{document}

\title{Error rates of Belavkin weighted quantum measurements and a converse to\\Holevo's asymptotic optimality Theorem}
\author{Jon Tyson\thanks{jonetyson@X.Y.Z, where X=post, Y=Harvard, Z=edu}\\Jefferson Laboratory, Harvard University, Cambridge MA 02138, USA}
\date{November 12, 2008\\
Journal Ref: \href{http://link.aps.org/doi/10.1103/PhysRevA.79.032343}{Phys.
Rev. A \textbf{79}, 032343 (2009)}.}
\maketitle

\begin{abstract}
\noindent We compare several instances of pure-state Belavkin weighted
square-root measurements from the standpoint of minimum-error discrimination
of quantum states. The quadratically weighted measurement is proven superior
to the so-called \textquotedblleft pretty good measurement\textquotedblright%
\ (PGM) in a number of respects:

\begin{enumerate}
\item Holevo's quadratic weighting unconditionally outperforms the PGM\ in the
case of two-state ensembles, with equality only in trivial cases.

\item A converse of a theorem of Holevo is proven, showing that a weighted
measurement is asymptotically-optimal only if it is quadratically weighted.

\end{enumerate}

\noindent Counter-examples for three states are constructed. The cube-weighted
measurement of Ballester\textit{, }Wehner, and Winter is also considered.
Sufficient optimality conditions for various weights are compared.

\end{abstract}

\pagebreak

\section{Introduction}

\subsection{Optimal measurements}

Consider an ensemble $\mathcal{E}_{m}^{\text{mixed}}$ of mixed quantum states
$\rho_{k}$ with \textit{a priori} probabilities $p_{k}$, $k=1,..,m,$ and unit
normalizations $\operatorname*{Tr}\rho_{k}=1$. Of fundamental importance is

\begin{quotation}
\noindent\textbf{The minimum-error quantum distinguishability problem:} If an
unknown state $\rho_{k}$ is blindly drawn from the ensemble, what is the
chance that the corresponding value of $k$ may be correctly identified by
performing an optimally chosen quantum measurement?
\end{quotation}

\noindent The modern approach to this problem is to consider measurements
defined by

\begin{definition}
A \textbf{positive-operator valued measure (POVM) }$\left\{  M_{k}\right\}  $
(see, for example, p. 74 of \cite{Helstrom Quantum Detection and Estimation
Theory}) is a collection of positive semidefinite operators on a Hilbert space
$\mathcal{H}$ such that $\sum M_{k}=%
\openone
$. The probability that the value $i$ is detected when the POVM is applied to
the state $\rho_{j}$ is given by $p_{i|j}=\operatorname*{Tr}M_{i}\rho_{j}$. In
particular, the \textbf{success rate} for the POVM to distinguish the ensemble
$\mathcal{E}_{m}^{\text{mixed}}$ is given by%
\begin{equation}
P_{\text{succ}}=\sum_{k=1}^{m}p_{k}\operatorname*{Tr}\left(  \rho_{k}%
M_{k}\right)  \text{.}%
\end{equation}

\end{definition}

Minimum-error quantum measurement was first considered in the 1960s in the
design of high performance optical detectors \cite{Helstrom Quantum Detection
and Estimation Theory}. More recently, this problem has been fundamentally
important in quantum Shannon theory (for example \cite{pure state HSW theorem,
mixed state HSW theorem,Holevo mixed state HSW theorem}) and in construction
of quantum algorithms for the \textit{Hidden Subgroup Problem}.\cite{Ip Shor's
algorithm is optimal, Bacon, Childs from optimal to efficient algo, optimal
alg for hidden shift, moore and russels distinguishing, Bacon new hidden
subgroup}. Various necessary and sufficient conditions for optimal
measurements have been derived \cite{Yuen Ken Lax Optimum testing of
multiple,Holevo optimal measurement conditions 1,Holevo remarks on optimal
measurements, Belavkin Optimal multiple quantum statistical hypothesis
testing, Belavkin and Vancjan, Barnett and Croke On the conditions for
discrimination between quantum states with minimum error} (see also
\cite{eldar short and sweet optimal measurements}). A number of relatively
recent works give interesting general upper and/or lower bounds on the quantum
distinguishability problem.
\cite
{pure state HSW theorem,Hayden Leung Multiparty Hiding,Barnum Knill UhOh,Montanaro on the distinguishability of random quantum states,
Daowen Qui Minimum-error discrimination between mixed quantum states,
Montanaro a lower bound on the probability of error in quantum state discrimination,
Qiu and Li Bounds on the minimum-error discrimination between mixed quantum states,QS2,Tyson Simplified and robust conditions for optimum testing of multiple hypotheses}
Explicitly solving the general optimal measurement problem is most likely
impossible, but in specific numerical cases one may compute the optimal
measurement by numerical iteration \cite{Tyson Simplified and robust
conditions for optimum testing of multiple hypotheses, Helstrom Bayes cost
reduction, Jezek Rehacek and Fiurasek Finding optimal strategies for minimum
error quantum state discrimination, Hradil et al Maximum Likelihood methods in
quantum mechanics} or by numerical solution of the associated semidefinite
program \cite{eldar short and sweet optimal measurements}.

The optimal measurement problem has been generalized to wave discrimination
\cite{Belavkin Book} and to optimal reversals of quantum channels, in the
sense of average entanglement fidelity \cite{Barnum Knill UhOh, Fletcher
Thesis Channel Adapted quantum error correction,Fletcher Shor Win Optimum
quantum error recovery using semidefinite programming,Fletcher Shor Win
Channel-Adapted quantum error correction for the amplitude dampin
channel,Fletcher Shor Win Structured Near-Optimal Channel-Adapted Quantum
Error Correction, Taghavi Channel-Optimized quantum error correction}. More
recently, the success-rate of optimal measurements has been expressed in terms
of the conditional min-entropy of corresponding classical-quantum states. (See
Theorem 1 of \cite{Konig Renner Schaffner Operational meaning of min and max
entropy}.)

\subsection{ Belavkin's Theorems}

In the rest of this paper we shall restrict consideration to the ensemble
\begin{equation}
\mathcal{E}_{m}=\left\{  \left(  \psi_{k},p_{k}\right)  \right\}  _{k=1,...,m}
\label{The pure state ensemble}%
\end{equation}
of pure quantum states $\psi_{k}\in\mathcal{H}$, and consider POVMs given by

\begin{definition}
The \textbf{Belavkin Weighted Square Root Measurement} (BWSRM) \cite{Belavkin
Optimal multiple quantum statistical hypothesis testing, Belavkin optimal
distinction of non-orthogonal quantum signals} (also known as a
\textbf{Weighted Least-Squares Measurement \cite{EldarSquareRootMeasurement})}
with weights $W_{k}\geq0$ is the POVM\footnote{The negative fractional power
is well-defined on the restriction to the span of the $W_{k}\psi_{k}$. More
properly, one may define $A^{-1/2}=%
{\displaystyle\sum}
\lambda_{k}^{-1/2}\left\vert \phi_{k}\right\rangle \left\langle \phi
_{k}\right\vert $, where $A=%
{\displaystyle\sum}
\lambda_{k}\left\vert \phi_{k}\right\rangle \left\langle \phi_{k}\right\vert $
is a spectral-decomposition with $\lambda_{k}=0$ terms omitted.}%
\begin{equation}
M_{k}=\left(  \sum_{\ell}W_{\ell}\left\vert \psi_{\ell}\right\rangle
\left\langle \psi_{\ell}\right\vert \right)  ^{-1/2}W_{k}\left\vert \psi
_{k}\right\rangle \left\langle \psi_{k}\right\vert \left(  \sum_{\ell}W_{\ell
}\left\vert \psi_{\ell}\right\rangle \left\langle \psi_{\ell}\right\vert
\right)  ^{-1/2}\text{.}
\label{eq formula for pure state weighted measurement M_k}%
\end{equation}
on the linear span of the $W_{k}\left\vert \psi_{k}\right\rangle $.
\end{definition}

The importance of BWSRMs in minimum-error discrimination problems was shown by
the following

\begin{theorem}
[Belavkin 1975 \cite{Belavkin Optimal multiple quantum statistical hypothesis
testing, Belavkin optimal distinction of non-orthogonal quantum signals}%
]\label{theorem belavkin pure state optimal weights}A POVM $\left\{
M_{k}\right\}  $ on $\operatorname*{Span}\left(  \mathcal{E}_{m}\right)
\equiv\operatorname*{Span}\left(  \left\{  p_{k}\psi_{k}\right\}  \right)  $
is optimal if and only if it may be expressed as a BWSRM with weights $W_{k}$
such that the operator%
\begin{equation}
\Lambda=\left(
{\displaystyle\sum_{\ell=1}^{m}}
W_{\ell}\left\vert \psi_{\ell}\right\rangle \left\langle \psi_{\ell
}\right\vert \right)  ^{1/2}
\label{In Belavkins theorem this better be invertible}%
\end{equation}
is invertible on $\operatorname*{Span}\left(  \mathcal{E}_{m}\right)  $ and%
\begin{equation}
p_{k}\left\langle \psi_{k}\right\vert \Lambda^{-1}\left\vert \psi
_{k}\right\rangle \leq1\text{,\label{Belavkin pure state optimality condition}%
}%
\end{equation}
with equality when $W_{k}>0$.\footnote{Mochon rediscovered that every optimal
pure-state measurement may be expressed as a BWSRM. \cite{Mochon PGM}}
\end{theorem}

Note that a simple formula for a set of optimal weights corresponding to a
given optimal measurement%
\begin{equation}
W_{k}^{\text{opt}}=\left\langle \psi_{k}\right\vert M_{k}^{\text{opt}%
}\left\vert \psi_{k}\right\rangle \times p_{k}^{2}%
\text{\label{equation for optimal weighting of a measurement}}%
\end{equation}
follows by squaring both sides of $\left(
\ref{Belavkin pure state optimality condition}\right)  $ and multiplying by
$W_{k}$, whether or not $W_{k}=0$. Furthermore, Belavkin's theorem implies
that optimal measurements on $\operatorname*{Span}\left(  \mathcal{E}%
_{m}\right)  $ satisfy $\operatorname{Rank}\left(  M_{k}^{\text{opt}}\right)
\leq1$.\footnote{For mixed states, $\operatorname{Rank}\left(  M_{k}%
^{\text{opt}}\right)  \leq\operatorname{Rank}\left(  \rho_{k}\right)  $
\cite{Belavkin Optimal multiple quantum statistical hypothesis testing}. See
also equation 5 of \cite{Jezek Rehacek and Fiurasek Finding optimal strategies
for minimum error quantum state discrimination} and \cite{eldar short and
sweet optimal measurements}. Equation
\ref{equation for optimal weighting of a measurement} may be understood
geometrically using \textquotedblleft frame forces,\textquotedblright\ which
have been advocated by Kebo and Benedetto \cite{Kebo thesis, frame force}.}

Belavkin and Maslov generalized Theorem
\ref{theorem belavkin pure state optimal weights} to mixed states (and more
generally to wave pattern recognition) in section 2.2 of \cite{Belavkin Book}.
Iteration of a mixed state version of equation $\left(
\ref{equation for optimal weighting of a measurement}\right)  $ was explored
in \cite{Jezek Rehacek and Fiurasek Finding optimal strategies for minimum
error quantum state discrimination, Hradil et al Maximum Likelihood methods in
quantum mechanics} as a method for numerical computation of optimal measurements.

Many of the known exactly solvable optimal pure state measurements are special
cases of

\begin{theorem}
[Belavkin 1975 \cite{Belavkin optimal distinction of non-orthogonal quantum
signals, Belavkin Optimal multiple quantum statistical hypothesis testing},
also Ban \cite{Ban Optimal signal detection in entanglement-assisted quantum
communication systems}]\label{theorem belavkin square root of P}The
measurement $\left(  \ref{eq formula for pure state weighted measurement M_k}%
\right)  $ for the weights $W_{k}=p_{k}$ is optimal for the pure state
ensemble $\mathcal{E}_{m}$ if the chance of successfully identifying a given
state $\psi_{k}$ is inversely proportional to its \textit{a priori}
probability:\footnote{The weighted measurement defined by $W_{k}=p_{k}$
sometimes appears as $M_{k}=\left\vert e_{k}\right\rangle \left\langle
e_{k}\right\vert $, with $\left\vert e_{k}\right\rangle =\sum_{\ell=1}%
^{m}\left\vert \psi_{\ell}\right\rangle \left(  P^{-1/2}\right)  _{\ell k}$,
where $P$ is the Graham matrix $P_{ij}=\sqrt{p_{i}p_{j}}\left\langle \psi
_{i},\psi_{j}\right\rangle $. Condition $\left(
\ref{condition for optimality of PGM in intuitive form}\right)  $ is then
written as $\left(  \sqrt{P}\right)  _{ii}=\left(  \sqrt{P}\right)  _{jj}$ for
all $i,j$. The equivalence of these two formulations follows from the matrix
identities $\Gamma^{\dag}\left(  \Gamma\Gamma^{\dag}\right)  ^{-1/2}=\left(
\Gamma^{\dag}\Gamma\right)  ^{-1/2}\Gamma^{\dag}$ and $\Gamma\left(
\Gamma^{\dag}\Gamma\right)  ^{-1/2}\Gamma^{\dag}=\left(  \Gamma\Gamma^{\dag
}\right)  ^{1/2}=\sqrt{P}$, where $\Gamma=\sum_{\ell=1}^{m}\sqrt{p_{\ell}%
}\left\vert \ell\right\rangle _{\mathbb{C}^{m}}\left\langle \psi
_{k}\right\vert _{\mathcal{H}}:\mathcal{H}\rightarrow\mathbb{C}^{m}$. Here
$\left\{  \left\vert \ell\right\rangle _{\mathbb{C}^{m}}\right\}  $ is the
standard orthonormal basis of $\mathbb{C}^{m}$.}%
\begin{equation}
p_{k}\left\langle \psi_{k}\right\vert M_{k}\left\vert \psi_{k}\right\rangle
=\text{const\label{condition for optimality of PGM in intuitive form}.}%
\end{equation}

\end{theorem}

\noindent Note that condition $\left(
\ref{condition for optimality of PGM in intuitive form}\right)  $ is
sufficient but not necessary, as may be seen by considering direct sums.
Belavkin originally applied Theorem $\ref{theorem belavkin square root of P}$
to homogeneous systems, cyclic systems, and systems of coherent
states.\cite{Belavkin optimal distinction of non-orthogonal quantum signals}
(A generalization of cyclic systems appears in \cite{Usuda et al Minimum error
detection of a classical linear code sending through a quantum channel}).

\subsection{Sub-optimal measurements}

In abstract studies of quantum channel capacities or quantum algorithms,
numerical routines for solving specific instances of optimal measurement
problem are often neither feasible nor desirable: one often has to rely on
sub-optimal measurements. Several extant approximately optimal measurements
are examples of

\begin{definition}
For $r>0$, the \textbf{Belavkin power-weighted square-root measurement
(BWSRM}-$r$\textbf{) }is the BWSRM with weights $W_{k}=p_{k}^{r}$.
\end{definition}

Examples of BWSRM-$r$'s appearing in the literature correspond to $r=1,2,3$.
Note that in the case of equiprobable ($p_{k}=1/m$) pure states that all
BWSRM-$r$'s are identical, and are of pervasive utility in quantum information
theory. (See, for example \cite{pure state HSW theorem}.)\footnote{The study
of BWSRM-$r$'s as approximately-optimal measurements in the equiprobable case
goes as far back as \cite{Helstrom Quantum Detection and Estimation Theory}
and \cite{Curlander thesis MIT}.}

We have already encountered the $r=1$ case in Theorem
\ref{theorem belavkin square root of P}. This measurement came to be known as
the \textquotedblleft pretty good measurement,\textquotedblright\ (PGM)
because of its reintroduction two decades later by Hausladen and Wootters as
an \textit{ad hoc approximately optimal measurement} \cite{HausladenThesis,
HausWootPGM} \textit{with simple error bounds}. Barnum and Knill showed that
the failure rate of the mixed-state version
\[
M_{k}^{\text{PGM}}=\left(  \sum p_{\ell}\rho_{\ell}\right)  ^{-1/2}p_{k}%
\rho_{k}\left(  \sum p_{\ell}\rho_{\ell}\right)  ^{-1/2}%
\]
of the Belavkin-Hausladen-Wootters PGM satisfies the bound
\begin{equation}
P_{\text{fail}}^{\text{opt}}\leq P_{\text{fail}}^{\text{PGM}}\leq
P_{\text{fail}}^{\text{opt}}\left(  1+P_{\text{succ}}^{\text{opt}}\right)
\leq2P_{\text{fail}}^{\text{opt}}\text{,\label{Barnum Knill Mixed Bound}}%
\end{equation}
where $P_{\text{fail}}^{\text{opt}}$ is the minimum-error failure rate.
\cite{Barnum Knill UhOh, Montanaro on the distinguishability of random quantum
states} The bound%
\begin{equation}
P_{\text{fail}}^{\text{PGM}}\leq%
{\displaystyle\sum_{i\neq j}}
p_{i}\left\vert \left\langle \psi_{i},\psi_{j}\right\rangle \right\vert ^{2}
\label{hayden scoop}%
\end{equation}
was proved by Hayden \textit{et al} \cite{Hayden Leung Multiparty Hiding},
generalizing the equiprobable bound of \cite{pure state HSW theorem}%
.\footnote{Equation $\left(  \ref{hayden scoop}\right)  $ follows by summing
the conditional error bound (A6) of \cite{Hayden Leung Multiparty Hiding}.
Bounds based on the pairwise quantities $\left\vert \left\langle \psi_{i}%
,\psi_{j}\right\rangle \right\vert ^{2}$ are inherently limited
\cite{Montanaro on the distinguishability of random quantum states}, although
frequently useful.}

The cube-weighted BWSRM-3 was employed by Ballester, Wehner, and Winter in the
study of state discrimination with post-measurement information.\cite{Wehner
thesis,Wehner State discrimination with post-measurement information}

\subsection{Asymptotically-optimal measurements \& BWSRM-2}

The quadratically weighted\ BWSRM-2 will be of particular interest in the
present work, and its mixed state generalization will be studied in the
sequel. Relatively recently, the mixed state version of BWSRM-2 has appeared
as the first iteration in a sequence of closed form measurements which appear
to converge to the optimal measurement.\cite{Jezek Rehacek and Fiurasek
Finding optimal strategies for minimum error quantum state discrimination,
Hradil et al Maximum Likelihood methods in quantum mechanics}. This weighting
was first specifically considered by Holevo, who was most interested in the
case of nearly orthogonal $\psi_{k}$.

\begin{definition}
\label{definition of asymptotic optimal measurement}A measurement procedure
$G$ for distinguishing the pure-state ensemble $\mathcal{E}_{m}$ is
\textbf{asymptotically optimal \cite{Holovo Assym Opt Hyp Test} }if for fixed
$p_{1},...,p_{m}$ one has
\[
\frac{P_{\text{fail}}^{\text{G}}\left(  \mathcal{E}_{m}\right)  }%
{P_{\text{fail}}^{\text{opt}}\left(  \mathcal{E}_{m}\right)  }\rightarrow1
\]
as the states $\psi_{k}$ approach an orthonormal basis.\footnote{It is
presumably intractable to produce a closed-form measurement process $G$ for
which $P_{\text{fail}}^{G}\left(  \mathcal{E}_{m}\right)  /P_{\text{fail}%
}^{\text{optimal}}\left(  \mathcal{E}_{m}\right)  \rightarrow1$ as the
$\psi_{k}$ and $p_{k}$ are arbitrarily varied in such a way that
$P_{\text{fail}}^{\text{optimal}}\left(  \mathcal{E}_{m}\right)  \rightarrow
0$. Otherwise, one could recover the optimal measurement for a fixed ensemble
$\mathcal{E}_{m}$ on $\mathcal{H}$ by taking the $\lambda\rightarrow1^{-}$
limit of the ensemble $\mathcal{E}_{m+1}^{\prime}\equiv\left\{  \left(
\psi_{k},\left(  1-\lambda\right)  p_{k}\right)  \right\}  \cup\left\{
\left(  \phi,\lambda\right)  \right\}  $ on a dilation $\mathcal{H}^{\prime
}\supset\mathcal{H}$, with $\phi\bot\mathcal{H}$.}
\end{definition}

\noindent Holevo showed that

\begin{theorem}
[Holevo's asymptotic-optimality Theorem (1977) \cite{Holovo Assym Opt Hyp
Test}]\label{Theorem Holevo's Assymtotic optimality theorem}The
quadratically-weighted pure state Belavkin measurement BWSRM-2 is
asymptotically optimal.
\end{theorem}

As we will see in section $\ref{section comparison of 2 pure state meas}$,
this property is \textbf{not} shared by the \textquotedblleft pretty good
measurement.\textquotedblright\ The key idea in Holevo's proof was the
construction of BWSRM-2 using an approximate minimal principle:

\begin{theorem}
[Holevo 1977 \cite{Holovo Assym Opt Hyp Test}.]%
\label{Theorem Holevo approx min principle}Assume that the states $\psi_{k}$
are linearly independent. Then the von Neumann measurement $M_{k}=\left\vert
e_{k}\right\rangle \left\langle e_{k}\right\vert $ minimizing%
\begin{equation}
C^{\text{Holevo}}\left(  \left\{  e_{k}\right\}  \right)  =\sum_{k=1}^{m}%
p_{k}\left\Vert \psi_{k}-e_{k}\right\Vert ^{2} \label{formula C Holevo}%
\end{equation}
over orthonormal\footnote{Orthogonal measurements are optimal for
distinguishing linearly-independent pure states.\cite{kennedy linearly
independent implies von Neumann,Belavkin optimal distinction of non-orthogonal
quantum signals,Helstrom Quantum Detection and Estimation Theory,Mochon PGM}}
sets $\left\{  e_{k}\right\}  $ is the quadratically weighted Belavkin
measurement BWSRM-$2$.
\end{theorem}

Theorem $\ref{Theorem Holevo approx min principle}$ was generalized by Eldar
and Forney \cite{EldarSquareRootMeasurement}, who showed that the BWSRM with
weights $W_{k}$ minimizes $C^{\left\{  W_{k}\right\}  }=\sum W_{k}\left\Vert
\psi_{k}-e_{k}\right\Vert ^{2}$ over POVMs $M_{k}=\left\vert e_{k}%
\right\rangle \left\langle e_{k}\right\vert $ without any assumption of linear
independence.\footnote{One can recover this generalization from Holevo's
argument using Naimark's Theorem. \cite{Kebo thesis} Note that the cost
function $C^{\text{Holevo}}$ for \textit{arbitrary }$p_{k}$ already appears as
eq. 8 of \cite{Holovo Assym Opt Hyp Test}.}

\subsection{Results}

In section \ref{asymptotic optimality section} it is shown that a weighted
measurement is asymptotically optimal only if it is quadratically weighted,
proving a converse to Holevo's asymptotic optimality theorem. In section
\ref{section comparison of 2 pure state meas} the PGM is found to be
categorically worse than the quadratically weighted measurement for two pure
states. In section \ref{section counterexamples for three states} we make a
heuristic comparison between various weightings, and present a counter-example
to show that the relationship between weightings is more complicated for
ensembles of more than two states. Finally, in section
\ref{section optimality conditions for weighted measurements} we compare
sufficient optimality conditions for various weightings.

\section{Pure State weighted measurements}

\subsection{Continuity\label{continuity section}}

Although weighted measurements are defined using the singular map $x\mapsto
x^{-1/2}$, one still has

\begin{theorem}
\label{Theorem continuity of WSRM}For fixed weights $W_{k}$, $k=1,...,m$, the
success rate of the weighted measurement for distinguishing the pure state
ensemble $\mathcal{E}_{m}=\left\{  \left(  \psi_{k},p_{k}\right)  \right\}
_{k=1,...,m}$ is a jointly continuous function of the $\psi_{k}$ and $p_{k}$.
\end{theorem}

\begin{proof}
Define the operator $A:\mathbb{C}^{m}\rightarrow\mathcal{H}$ by%
\begin{equation}
A=\sum_{k=1}^{m}\sqrt{W_{k}}\left\vert \psi_{k}\right\rangle _{\mathcal{H}%
}~\left\langle k\right\vert _{\mathbb{C}^{m}}%
\text{,\label{equation convert weightings to operator}}%
\end{equation}
where $\left\{  \left\vert k\right\rangle _{\mathbb{C}^{m}}\right\}  $ is the
standard orthonormal basis of $\mathbb{C}^{m}$. Then%
\begin{align}
P_{\text{succ}}^{W\text{-weighted}}  &  =\sum_{k=1}^{m}p_{k}\left\vert
\left\langle \psi_{k}\right\vert \left(  \sum_{\ell=1}^{m}W_{\ell}\left\vert
\psi_{\ell}\right\rangle \left\langle \psi_{\ell}\right\vert \right)
^{-1/2}W_{k}^{1/2}\left\vert \psi_{k}\right\rangle \right\vert ^{2}\nonumber\\
&  =\sum_{k=1}^{m}\frac{p_{k}}{W_{k}}\left\vert \left\langle k\right\vert
_{\mathbb{C}^{m}}A^{\dag}\left(  AA^{\dag}\right)  ^{-1/2}A\left\vert
k\right\rangle _{\mathbb{C}^{m}}\right\vert ^{2}\nonumber\\
&  =\sum_{k=1}^{m}\frac{p_{k}}{W_{k}}\left(  \left\langle k\right\vert \left(
A^{\dag}A\right)  ^{+1/2}\left\vert k\right\rangle \right)  ^{2}.
\label{expression to show WLSM ER is cont}%
\end{align}
Continuity of $P_{\text{fail}}$ follows from the continuity of the square
root.\footnote{By the Weierstrauss approximation theorem \cite{baby rudin},
given $\varepsilon>0$ one can find a polynomial $P$ such that $\left\vert
P\left(  \lambda\right)  -\sqrt{\lambda}\right\vert <\varepsilon$ for all
$\lambda$ in the interval $I=\left[  0,\sum W_{k}\right]  .$ Since $I$
contains the spectrum of $A^{\dag}A$ for any choice of $\left\{  \psi
_{k}\right\}  $, continuity of $\left(
\ref{expression to show WLSM ER is cont}\right)  $ is guaranteed by Theorem
7.12 of \cite{baby rudin}.}
\end{proof}

\newpage

\subsection{\label{section comparison of 2 pure state meas}Explicit comparison
of weighted Belavkin measurements for 2 pure states}

We first consider binary ensembles:

\begin{theorem}
\label{Theorem success prob of weighted 2-d measurements}The failure rates for
distinguishing the binary ensemble $\mathcal{E}_{2}$ using optimal and
weighted measurements are given by%
\begin{align}
P_{\text{fail}}^{\text{optimal}}  &  =\frac{1}{2}-\sqrt{\frac{1}{4}-p_{1}%
p_{2}\left\vert \left\langle \psi_{1},\psi_{2}\right\rangle \right\vert ^{2}%
}\label{formula for optimal 2 pure state failure rate}\\
P_{\text{fail}}^{\text{weighted}}  &  =\frac{\left(  p_{1}W_{2}+p_{2}%
W_{1}\right)  \cos^{2}\theta}{W_{1}+W_{2}+2\sqrt{W_{1}W_{2}}\left\vert
\sin\theta\right\vert }, \label{formula for W-weighted failure rate}%
\end{align}
where $\cos\theta=\left\vert \left\langle \psi_{1},\psi_{2}\right\rangle
\right\vert $.
\end{theorem}

\begin{proof}
Equation $\left(  \ref{formula for optimal 2 pure state failure rate}\right)
$ is equation 2.34 on page 113 of \cite{Helstrom Quantum Detection and
Estimation Theory}.

For an arbitrary $2\times2$ positive matrix $B$ it is easy to use the spectral
theorem to verify that%
\[
B^{+1/2}=\left(  2\sqrt{\det B}+\operatorname*{Tr}B\right)  ^{-1/2}\left(
B+\sqrt{\det B}\times%
\openone
\right)  .
\]
For $A$ defined by $\left(  \ref{equation convert weightings to operator}%
\right)  $ one has%
\begin{align*}
\det A^{\dag}A  &  =W_{1}W_{2}\sin^{2}\theta\\
\operatorname*{Tr}A^{\dag}A  &  =W_{1}+W_{2}.
\end{align*}
Equation $\left(  \ref{formula for W-weighted failure rate}\right)  $ now
follows from $\left(  \ref{expression to show WLSM ER is cont}\right)  $:%
\begin{align*}
P_{\text{succeed}}^{W_{1},W_{2}}  &  =\sum\frac{p_{k}}{W_{k}}\left\vert
\left\langle k\right\vert \left(  A^{\dag}A\right)  ^{+1/2}\left\vert
k\right\rangle \right\vert ^{2}\\
&  =\sum\frac{p_{k}}{W_{k}}\left\vert \left(  2\sqrt{W_{1}W_{2}}\left\vert
\sin\theta\right\vert +W_{1}+W_{2}\right)  ^{-1/2}\left(  W_{k}+\sqrt
{W_{1}W_{2}}\left\vert \sin\theta\right\vert \right)  \right\vert ^{2}\\
&  =1-\frac{\sum p_{k}W_{1-k}\cos^{2}\theta}{W_{1}+W_{2}+2\sqrt{W_{1}W_{2}%
}\left\vert \sin\theta\right\vert }.
\end{align*}

\end{proof}

\begin{theorem}
[Holevo's measurement is better than the PGM\ for two pure states]For
distinguishing the $2$-pure-state ensemble $\mathcal{E}_{2}$ one has the
following inequalities
\begin{align}
P_{\text{fail}}^{\text{PGM}}  &  \geq P_{\text{fail}}^{\text{Holevo}%
},\label{eq PGM worse HOM for 2 pure states}\\
2=\sup_{\mathcal{E}_{2}}\frac{P_{\text{fail}}^{\text{PGM}}}{P_{\text{fail}%
}^{\text{optimal}}}  &  >\sup_{\mathcal{E}_{2}}\frac{P_{\text{fail}%
}^{\text{Holevo}}}{P_{\text{fail}}^{\text{optimal}}}=\frac{\sqrt{2}+1}%
{2}\approx1.207\label{supremum error for two dimensions}\\
&  >\sup_{\mathcal{E}_{2}}\frac{P_{\text{fail}}^{\text{cubic weighting}}%
}{P_{\text{fail}}^{\text{optimal }}}\approx1.118,
\label{numeric value for cubic
failure on two-state ensembles}%
\end{align}
with equality in $\left(  \ref{eq PGM worse HOM for 2 pure states}\right)  $
iff $p_{1},p_{2}\in\left\{  0,1/2,1\right\}  $ or $\left\langle \psi_{1}%
,\psi_{2}\right\rangle =0$.
\end{theorem}

\begin{proof}
To prove $\left(  \ref{eq PGM worse HOM for 2 pure states}\right)  $, note
that since $\sqrt{p_{1}p_{2}}\leq\frac{1}{2},$ we have the inequalities%
\begin{align*}
\frac{1}{2}+\sqrt{p_{1}p_{2}}  &  \leq1\\
\sqrt{p_{1}p_{2}}\left(  \left\vert \sin\theta\right\vert -1\right)   &
\leq2p_{1}p_{2}\left(  \left\vert \sin\theta\right\vert -1\right)  .
\end{align*}
Summing gives
\[
\frac{1}{2}\left(  1+2\sqrt{p_{1}p_{2}}\left\vert \sin\theta\right\vert
\right)  \leq1+2p_{1}p_{2}\left(  \left\vert \sin\theta\right\vert -1\right)
=p_{1}^{2}+p_{2}^{2}+2p_{1}p_{2}\sin\left\vert \theta\right\vert .
\]
Equation $\left(  \ref{eq PGM worse HOM for 2 pure states}\right)  $ follows
by dividing $p_{1}p_{2}\cos^{2}\theta$ by both sides and applying $\left(
\ref{formula for W-weighted failure rate}\right)  $.

The equation on the left-hand side of $\left(
\ref{supremum error for two dimensions}\right)  $ shows that the bound
$\left(  \ref{Barnum Knill Mixed Bound}\right)  $ of Barnum and Knill is
sharp. To see that%
\[
\sup_{\mathcal{E}_{2}}\frac{P_{\text{fail}}^{\text{PGM}}}{P_{\text{fail}%
}^{\text{optimal}}}\geq2,
\]
take $p_{1}\rightarrow0^{+}$ for any fixed $\left\langle \psi_{1},\psi
_{2}\right\rangle \neq0$. The equation on the RHS of $\left(
\ref{supremum error for two dimensions}\right)  $ is an unilluminating
exercise in calculus. The maximizing ensemble is given by $\psi_{1}=\psi_{2}$
and $p_{1}=\sqrt{2}/2$. The last inequality $\left(
\ref{numeric value for cubic failure on two-state ensembles}\right)  $ was
computed numerically.
\end{proof}

The relative success rates of $P_{\text{fail}}/P_{\text{fail}}^{\text{optimal}%
}$ for the weightings $r=1,2,$ and $3$ of measurements on the ensemble
$\mathcal{E}_{2}$ with $\left\vert \left\langle \psi_{1},\psi_{2}\right\rangle
\right\vert =\cos\theta$ and $p_{1}=1-p_{2}=p$ are plotted in Figures 1 and 2a/b.%

\begin{figure}
[ptb]
\begin{center}
\includegraphics[
trim=0.000000in 0.399346in 0.000000in 0.369765in,
height=2.6195in,
width=3.3892in
]%
{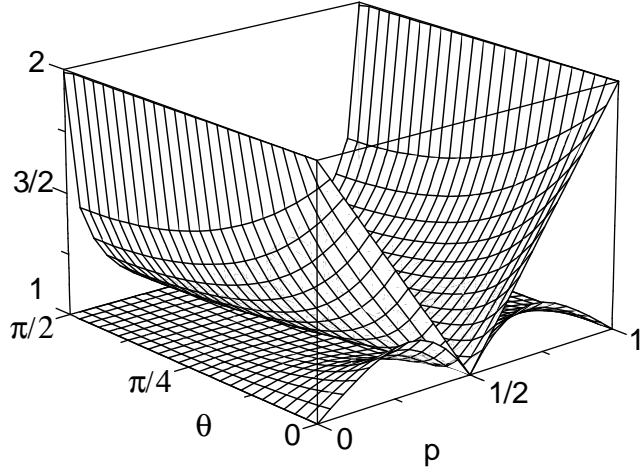}%
\caption{$%
\protect\begin{tabular}
[t]{l}%
$P_{\text{fail}}/P_{\text{fail}}^{\text{opt}}$ for the PGM (upper) and
Holevo's \protect\\
measurement (lower) for binary ensembles\protect\\
with $\left\vert \left\langle \psi_{1},\psi_{2}\right\rangle \right\vert
=\cos\theta$ and $p_{1}=1-p_{2}=p$%
\protect\end{tabular}
\ $}%
\end{center}
\end{figure}
\[%
\begin{tabular}
[c]{ll}%
{\parbox[b]{2.885in}{\begin{center}
\includegraphics[
trim=0.000000in 0.401699in 0.000000in 0.366067in,
height=2.2329in,
width=2.885in
]%
{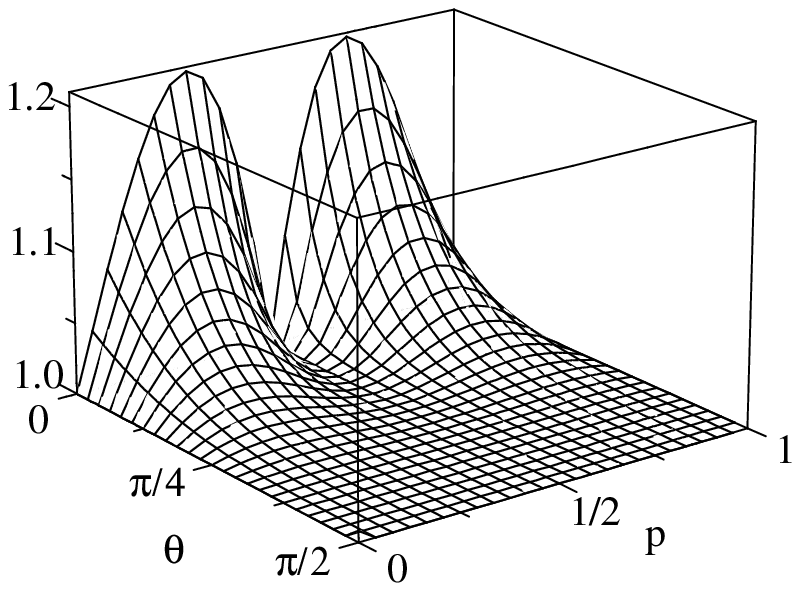}%
\\
Fig 2(a): $P_{\text{fail}}^{\text{Holevo}}/P_{\text{fail}}^{\text{opt}}$ for
binary ensembles
\end{center}}}%
&
{\parbox[b]{2.8971in}{\begin{center}
\includegraphics[
trim=0.000000in 0.401699in 0.000000in 0.366067in,
height=2.2312in,
width=2.8971in
]%
{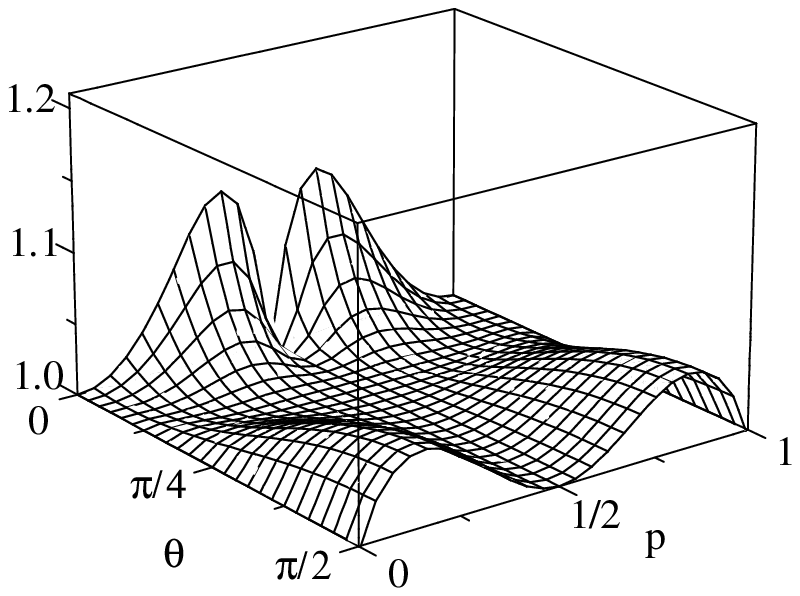}%
\\
Fig 2(b): $P_{\text{fail}}^{\text{cubic}}/P_{\text{fail}}^{\text{opt}}$ for
binary ensembles
\end{center}}}%
\end{tabular}
\ \ \ \
\]

\subsection{Asymptotic optimality\label{asymptotic optimality section}}

Holevo's quadratic weighting is uniquely characterized by the following
converse of Theorem \ref{Theorem Holevo's Assymtotic optimality theorem}:

\begin{theorem}
[Converse to Holevo's asymptotic optimality Theorem]Fix probabilities
$p_{k}>0$ and weights $W_{k}\geq0$, $k=1,...,m$. Then the Belavkin weighted
square root measurement $\left(
\ref{eq formula for pure state weighted measurement M_k}\right)  $ is
asymptotically optimal for distinguishing ensembles $\mathcal{E}_{m}=\left\{
\left(  \psi_{k},p_{k}\right)  \right\}  _{k=1,...,m}$ only if $W_{k}%
=$const$\times p_{k}^{2}$.
\end{theorem}

\begin{proof}
Setting $c_{k}=W_{k}/p_{k}^{2}$, we must show that $c_{k}=c_{k^{\prime}}$ for
all $k,k^{\prime}$ under the assumption that $\left\{  W_{k}\right\}  $
defines an asymptotically optimal measurement. It is sufficient to consider
the case $m=2$.\footnote{The case $m>2$ is reduced to $m=2$ by considering
ensembles for which each of a subset $m-2$ states is orthogonal to all of the
other states in $\mathcal{E}_{m}$.} By Theorem
$\ref{Theorem success prob of weighted 2-d measurements}$ and L'Hospital's
rule%
\begin{align*}
\lim_{\theta\rightarrow\pi/2}\frac{P_{\text{fail}}^{\text{weighted}}%
}{P_{\text{fail}}^{\text{optimal}}}  &  =\frac{p_{1}W_{2}+p_{2}W_{1}}%
{W_{1}+W_{2}+2\sqrt{W_{1}W_{2}}}\times\lim_{\theta\rightarrow\pi/2}\frac
{\cos^{2}\theta}{\frac{1}{2}-\sqrt{\frac{1}{4}-p_{1}p_{2}\cos^{2}\theta}}\\
&  =\frac{\left(  c_{2}p_{2}+c_{1}p_{1}\right)  p_{1}p_{2}}{\left(
\sqrt{W_{1}}+\sqrt{W_{2}}\right)  ^{2}}\times\frac{1}{p_{1}p_{2}}\\
&  =\frac{c_{1}p_{1}+c_{2}p_{2}}{\left(  \sqrt{c_{1}}p_{1}+\sqrt{c_{2}}%
p_{2}\right)  ^{2}}\text{.}%
\end{align*}
The conclusion follows from the strict convexity of $x\mapsto x^{2}$.
\end{proof}

\subsection{Reflections \& Counter-examples for three
states\label{section counterexamples for three states}}

We now reflect on the relationships between Belavkin's optimal weighting
($W_{k}=p_{k}^{2}\left\langle \psi_{k}\right\vert M_{k}^{\text{opt}}\left\vert
\psi_{k}\right\rangle $), the weighting for the PGM ($W_{k}=p_{k}$), Holevo's
weighting ($W_{k}=p_{k}^{2}$), and the weighting of Ballester, Wehner, and
Winter ($W_{k}=p_{k}^{3}$). Note that while Holevo's measurement relatively
over-weights vectors $\psi_{k}$ for which $\left\langle \psi_{k}\right\vert
M_{k}^{\text{opt}}\left\vert \psi_{k}\right\rangle $ is relatively small, the
PGM\ additionally over-weights vectors for which $p_{k}$ is small! In general,
one therefore expects that \textit{the relative misweightings of the PGM\ tend
to compound one another, so that BWSRM-2 is better that BWSRM-1}. Similarly,
by approximate cancellation of misweightings, one expects that the cubic
weighting will sometimes outperform the quadratic weighting for ensembles far
from the asymptotically orthogonal regime considered by Holevo.

We have seen that Holevo's measurement is always as least as good as the
PGM\ for two-state ensembles. For three states the above intuitive argument
does not always hold true, as shown by the following pathology:

\begin{theorem}
\label{Theorem bad pure state example}There exists a 3-state ensemble with the
properties that:

\begin{enumerate}
\item \label{property a priori most probable never detected}There is an
optimal measurement such that the a priori strictly-most-probable state is
NEVER detected.

\item \label{holevo worse than PGM}Holevo's measurement is worse than the PGM:
$P_{\text{fail}}^{\text{Holevo}}>P_{\text{fail}}^{\text{PGM}}$.
\end{enumerate}
\end{theorem}

\begin{proof}
Define the $3$-state ensemble by $\psi_{1}=\left(  \cos\theta,\sin
\theta\right)  $, $\psi_{2}=\left(  \cos\theta,-\sin\theta\right)  $, and
$\psi_{3}=\left(  1,0\right)  ,$ where $\theta=\pi/6$ and%
\[
p_{1}=p_{2}=\left(  1-p_{3}\right)  /2=\left(  2+\left(  \cos\theta+\sin
\theta\right)  \cos\theta\right)  ^{-1}\approx.3142<1/3\text{.}%
\]
It is straightforward to check that the POVM%
\[%
\begin{array}
[c]{ccc}%
M_{1}=\frac{1}{2}%
\begin{bmatrix}
1 & 1\\
1 & 1
\end{bmatrix}
& M_{2}=\frac{1}{2}%
\begin{bmatrix}
1 & -1\\
-1 & 1
\end{bmatrix}
& M_{3}=0
\end{array}
\]
satisfies the necessary and sufficient optimality conditions \cite{Holevo
remarks on optimal measurements}
\[
\left(  L+L^{\dag}\right)  /2-p_{k}\left\vert \psi_{k}\right\rangle
\left\langle \psi_{k}\right\vert \geq0\text{ for all }k\text{,}%
\]
where the Lagrange operator $L$ is given by%
\[
L\equiv\sum_{k=1}^{3}p_{k}\,M_{k}\left\vert \psi_{k}\right\rangle \left\langle
\psi_{k}\right\vert =p_{1}\left(  \cos\left(  \theta\right)  +\sin\left(
\theta\right)  \right)
\begin{bmatrix}
\cos\theta & \\
& \sin\theta
\end{bmatrix}
\text{.}%
\]
(Here $A\geq0$ means $A$ is positive semidefinite.) Property
$\ref{holevo worse than PGM}$ follows by direct computation:%
\[
P_{\text{fail}}^{\text{Holevo}}\approx.4245>\text{\ }P_{\text{fail}%
}^{\text{PGM}}\approx.4224>\ P_{\text{fail}}^{\text{optimal}}\approx.4138.
\]

\end{proof}

\textbf{Remark:} The linear-dependence of the states in the above construction
was not essential: one can simply embed the above example in $3$-space, and
perturb the vectors $\psi_{k}$ slightly to make them linearly independent. By
Theorem $\ref{Theorem continuity of WSRM}$, property
$\ref{holevo worse than PGM}$ will be unaffected by small perturbations.

Theorem \ref{Theorem bad pure state example} aside, given any fixed set of
non-equal priors $p_{k}$, we conjecture that Holevo's weighting will have a
better success rate than the PGM on average for randomly chosen ensembles
$\mathcal{E}_{m}$, with the corresponding $\left\{  \psi_{k}\right\}  $
independently chosen according to Haar measure.

\subsection{Sufficient optimality conditions for weighted
measurements\label{section optimality conditions for weighted measurements}}

We close our comparisons of Belavkin weighted measurements by noting that in
the case $W_{k}>0$, Theorem $\ref{theorem belavkin square root of P}$
generalizes easily:

\begin{theorem}
[Optimality conditions for positively-weighted measurements]%
\label{theorem weighted measurement optimality conditions}A sufficient
condition for optimality of the BWSRM with strictly positive weights $W_{k}>0$
is that there exists a constant $c>0$ such that
\begin{equation}
p_{k}^{2}\left\langle \psi_{k}\right\vert M_{k}\left\vert \psi_{k}%
\right\rangle =cW_{k}\text{ for all }%
k\text{.\label{eq in terms of weighted meas for optimality cond}}%
\end{equation}

\end{theorem}

\begin{proof}
Dividing both sides of $\left(
\ref{eq in terms of weighted meas for optimality cond}\right)  $ by $cW_{k}$
and taking the square root gives
\[
c^{-1/2}p_{k}\left\langle \psi_{k}\right\vert \left(
{\displaystyle\sum}
W_{\ell}\left\vert \psi_{\ell}\right\rangle \left\langle \psi_{\ell
}\right\vert \right)  ^{-1/2}\left\vert \psi_{k}\right\rangle =1.
\]
In particular, the rescaled weights $c\times W_{\ell}$ satisfy Belavkin's
optimality condition $\left(  \ref{Belavkin pure state optimality condition}%
\right)  $. The result follows, since BWSRMs are unaffected by such rescalings.
\end{proof}

The assumption that $W_{k}>0$ for all $k$ is necessary:\ otherwise the weights
$W_{k}=\delta_{k1}$ would be optimal for any ensemble. Note that for the
asymptotically optimal weight $W_{k}=p_{k}^{2}$, equation $\left(
\ref{eq in terms of weighted meas for optimality cond}\right)  $ becomes the
particularly simple condition
\begin{equation}
\left\langle \psi_{k}\right\vert M_{k}\left\vert \psi_{k}\right\rangle
=\text{const.}%
\end{equation}

\section{Future directions}

In the sequel, we focus on the quadratically weighted mixed state measurement,
and consider resulting two-sided bounds for the distinguishability arbitrary
ensembles of mixed quantum states \cite{sequel}. We will generalize to the
case of approximate reversals of quantum channels at a later date \cite{In
preparation}.\bigskip

\noindent\textbf{Acknowledgements:} I would like to thank Andrew Childs, Aram
Harrow, Julio Concha, V. P. Belavkin, and Vincent Poor for pointing out useful
references, Julio Concha and Andrew Kebo for providing copies of their theses,
William Wootters for providing a copy of Hausladen's thesis, and Arthur Jaffe
for his encouragement.

\end{document}